\documentclass[conference, twocolumn]{IEEEtran}
\usepackage{cite}
\usepackage{amsmath,amssymb,amsfonts}
\usepackage{algorithmic}
\usepackage{textcomp} 
\usepackage{xcolor}
\usepackage{csquotes}
\usepackage{makecell}
\usepackage{color}
\usepackage{multirow}%
\usepackage{comment}
\usepackage{hyperref}
\usepackage{url}

\newtheorem{definition}{\rm\textbf{Definition}}
\newtheorem{theorem}{\rm\textbf{Theorem}}
\newtheorem{lemma}{\rm\textbf{Lemma}}
\newtheorem{corollary}{\rm\textbf{Corollary}}

\newtheorem{assumption}{\rm\textbf{Assumption}}

\ifCLASSINFOpdf
\usepackage[pdftex]{graphicx}
\else
\usepackage[dvips]{graphicx}
\fi

\begin{document}
\title{Secure Control of Connected and Automated Vehicles Using Trust-Aware Robust Event-Triggered  Control Barrier Functions}



\makeatletter
\newcommand{\linebreakand}{%
  \end{@IEEEauthorhalign}
  \hfill\mbox{}\par
  \mbox{}\hfill\begin{@IEEEauthorhalign}
}
\makeatother

\author{\IEEEauthorblockN{H M Sabbir Ahmad}
\IEEEauthorblockA{Boston University\\
sabbir92@bu.edu}
\and
\IEEEauthorblockN{Ehsan Sabouni}
\IEEEauthorblockA{Boston University\\
esabouni@bu.edu}
\and
\IEEEauthorblockN{Akua Dickson}
\IEEEauthorblockA{Boston University\\
akuad@bu.edu}
\linebreakand % <------------- \and with a line-break
\IEEEauthorblockN{Wei Xiao}
\IEEEauthorblockA{Massachusetts Institute of Technology\\
weixy@mit.edu}
\and
\IEEEauthorblockN{Christos G. Cassandras}
\IEEEauthorblockA{Boston University\\
cgc@bu.edu}
\and
\IEEEauthorblockN{Wenchao Li}
\IEEEauthorblockA{Boston University\\
wenchao@bu.edu}
\thanks{This work was supported in part by NSF under grants CPS-1932162, ECCS-1931600, DMS-1664644, CNS-2149511, and by ARPA-E under grant DE-AR0001282.}}

\IEEEoverridecommandlockouts
\makeatletter\def\@IEEEpubidpullup{6.5\baselineskip}\makeatother
\IEEEpubid{\parbox{\columnwidth}{
    Symposium on Vehicles Security and Privacy (VehicleSec) 2024 \\
    26 February 2024, San Diego, CA, USA \\
    ISBN 979-8-9894372-7-6 \\
    https://dx.doi.org/10.14722/vehiclesec.2024.23037 \\
    www.ndss-symposium.org
}
\hspace{\columnsep}\makebox[\columnwidth]{}}

\maketitle

\begin{abstract}

We address the security of a network of Connected and Automated Vehicles (CAVs) cooperating to safely navigate through a conflict area 
(e.g., traffic intersections, merging roadways, roundabouts). Previous studies have shown that such a network can be targeted by adversarial attacks causing traffic jams or safety violations ending in collisions. 
We focus on attacks targeting the V2X communication network used to share vehicle data and consider as well uncertainties due to noise in sensor measurements and communication channels. To combat these, motivated by recent work on the safe control of CAVs, we propose a trust-aware robust event-triggered decentralized control and coordination framework that can provably guarantee safety. 
We maintain a trust metric for each vehicle in the network computed based on their behavior and used to balance the tradeoff between conservativeness (when deeming every vehicle as untrustworthy) and guaranteed safety and security.
It is important to highlight that our framework is invariant to the specific choice of the trust framework. 
Based on this framework, we propose an attack detection and mitigation scheme which has twofold benefits: (i) the trust framework is immune to false positives, and (ii) it provably guarantees safety against false positive cases.
We use extensive simulations (in SUMO and CARLA) to validate the theoretical guarantees and demonstrate the efficacy of our proposed scheme to detect and mitigate adversarial attacks. The code for the simulated scenarios can be found in this \href{https://github.com/SabbirAhmad26/Trust_based_CBF}{\textit{\underline{link}}}.
\end{abstract}

\section{Introduction}
\label{Introduction}
The emergence of Connected and Automated Vehicles (CAVs) and advancements in traffic infrastructure \cite{li2013survey} promise to offer solutions to transportation issues like accidents, congestion, energy consumption, and pollution \cite{deWaard09,kavalchuk2020performance}. To achieve these benefits, secure and efficient traffic management is crucial, particularly at bottleneck locations such as intersections, roundabouts, and merging roadways \cite{VANDENBERG201643}. 

We focus on \emph{decentralized} algorithms as they provide manifold benefits, including added security since an attacker can only target a limited number of agents; in contrast, in a centralized scheme an attack on the central entity can potentially compromise every agent/CAV. Security of Autonomous Vehicles (AVs) has been extensively studied in existing literature \cite{Shukla_01,Sun_01,Pham_01} whereby the attacks can be broadly categorized into in-vehicle network attacks and V2V or V2X communication network attacks. 
There has been significant research done \cite{xu2019grouping,Xu_03,Xiao_03} from a control point of view with the aim of designing efficient real-time controllers for CAVs. However, ensuring security in the implementation of these controllers has received little attention, with the literature mostly limited to the security of Cooperative Adaptive Cruise Control (CACC) \cite{Lu_01,Biroon_02,Boddupalli_01,Alipour_01,Farivar_01}. 
These studies do not extend to the more critical parts of a traffic network such as intersections or roundabouts, where the repercussions of an attack are more severe, yet the literature addressing security in these cases is limited. 

The authors in \cite{Jarouf_01} propose a technique based on public key cryptography, while \cite{Zhao_01} assesses cybersecurity risks on cooperative ramp merging by targeting V2I communication with road-side units (RSU). 
More comprehensive studies of the security of decentralized control and coordination algorithms for CAVs can be found in \cite{ahmad_01, ahmad_02}. 
%One of the studies highlights the attack surface, cyber threats, and their impact on the CAV network. 
In \cite{ahmad_02}, an attack resilient control and coordination algorithm is proposed using Control Barrier Functions (CBFs) without any mitigation technique. Moreover, the framework in \cite{ahmad_02} only uses V2X communication without local perception, which we deem highly useful for added security. It is also not robust to uncertainties in state estimates/measurements, which poses a security limitation as many stealthy attacks are designed to evade detection by a Bad Data Detector (BDD).
%CACC extends the idea of Adaptive Cruise Control (ACC) by incorporating cooperation through Vehicle-to-vehicle (V2V) through Dedicated Short Range Communication (DSRC) which makes CACC vulnerable to cyberattacks. Additionally, in-vehicular network attacks can also imperil CACC performance. 

The notion of trust/reputation has been applied to multi-agent systems including Intelligent Transportation Systems (ITS) in \cite{Chen_01, Cheng_02, Hu_01, Hu_02}.
%using subjective logic, \cite{Hu_01, Hu_02} had adopted Beta reputation systems, and \cite{Josang_01} applying Dirichlet reputation systems. 
%However, these works do not address the security of the coordination algorithms for CAVs. 
In \cite{Parwana_01} a novel trust-based CBF framework is proposed for multi-robot systems (MRSs) to provide safe control against adversarial agents. However, this cannot be directly applied to a traffic network as it is limited to a specific characterization of agents that does not apply to a road network. The authors in \cite{Garlichs_01} propose a trust framework to address the security of CACC. Along a similar vein, \cite{Shoukry_01} employs a trust framework to address Sybil attacks within traffic intersections using a macroscopic network model. However, it is constrained by the accuracy of the traffic density estimation model in detecting fake vehicles (CAVs) and also offers no guarantees on preventing false positives (i.e., detecting all fake vehicles accurately and not detecting any real vehicle as fake).  

In our paper we tackle the aforementioned shortcomings beginning with making the control and coordination robust against disturbances and uncertainties in the states measurements and estimations. Besides that, we also incorporate mitigation in order to subside the effect of the attacks on the network performance. Thus our main contributions are summarized below:

\begin{enumerate}
    \item We propose a \textit{novel robust trust-aware event-triggered control and coordination framework} that guarantees safe coordination for CAVs in conflict areas in the presence of adversarial attacks. 
    \item Our proposed formulation is robust against stealthy attacks that can pass through BDDs undetected. The benefit of event-triggered control lies in reducing the communication load, thus improving robustness against attacks.
    \item We propose an \textit{attack detection and mitigation scheme based on the trust score of CAVs} that can alleviate the effect of the attack, particularly the case of traffic holdup by restoring normal coordination. Our proposed scheme guarantees safety against false positive (FP) cases, which may arise due to a poor choice (or, design) of the trust framework.
\end{enumerate}

While our framework views security as a specification in a control and coordination problem, it is important to note that various network security measures like cryptographic techniques can complement this framework. The paper is organized in six sections. The next section provides some background, followed by the threat model in Section~\ref{threat_model}. In Section~\ref{control_coordination}, we present the robust event-triggered control and coordination framework, which is followed by the attack mitigation in Section~\ref{mitigation}. We present simulation results in Section~\ref{results}. Finally, the conclusion is included in Section~\ref{conclusion}.

\section{Background} 
\label{problem_formulation}
We present a resilient control and coordination approach that includes an attack detection and mitigation scheme for secure coordination of CAVs in conflict areas using the signal-free intersection presented in \cite{xu2021comparison} as an illustrative example. Figure \ref{fig:intersection} shows a typical intersection with multiple lanes. Here, the Control Zone (CZ) is the area within the circle. 
containing eight entry lanes labeled from $O_1$ to $O_8$ and exit lanes labeled from $l_1$ to $l_8$ %\li{mix of capitalized and non-capitalized symbols} 
each of length $L$ which is assumed to be the same here. Red dots show all the merging points (MPs)%\li{acronyms introduced before?}
where potential collisions may occur. All the CAVs have the following possible movements: going straight, turning left from the leftmost lane, or turning right from the rightmost lane. 

%{\color{red}Any reason for excluding the objective function here?}
\begin{figure*}
\centering
\includegraphics[scale = 0.6]{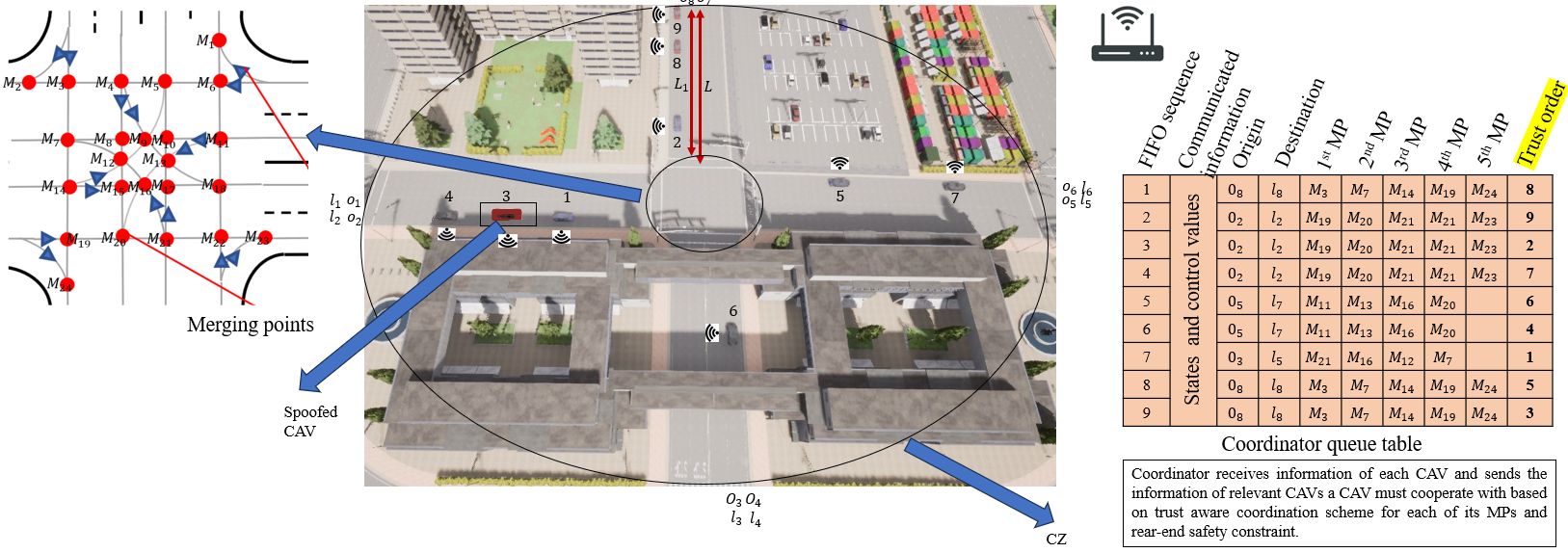}
\caption{The multi-lane intersection problem. Collisions may happen at the MPs (red dots shown in above figure). 
}
\label{fig:intersection}%
\end{figure*}

The vehicle dynamics for each CAV in the CZ take the following form
\begin{equation}
\left[
\begin{array}
[c]{c}%
\dot{x}_{i}(t)\\
\dot{v}_{i}(t)
\end{array}
\right] =\left[
\begin{array}
[c]{c}%
v_{i}(t)\\
u_{i}(t)
\end{array}
\right], \label{VehicleDynamics}%
\end{equation}
where $x_{i}(t)$ is the distance along the lane from the origin at which CAV $i$ arrives, $v_{i}(t)$ and $u_{i}(t)$ denote the velocity and control input (acceleration/deceleration) of CAV $i$, respectively.

A road-side unit (RSU) acts as a coordinator which receives and stores the state and control information $[x_i(t), v_i(t), u_i(t)]^T$ from CAVs through vehicle-to-infrastructure (V2X) communication. Additionally, it also stores and updates the a trust metric for each CAV in the CZ. It is assumed that the coordinator knows the entry and exit lanes for each CAV upon their arrival and uses it to determine
%\li{trajectory can be a bit confusing here; perhaps qualify it as planned trajectory and then elaborate further on what it entails} 
the list of MPs in its planned trajectory. It facilitates safe coordination by providing each CAV with relevant information about other CAVs in the network, particularly those that are at risk of collision.

\subsection{Constraints/rules in the Control Zone}
\label{rules}
Let $t_{i}^{0}$ and $t_{i}^{f}$ denote the time that CAV $i$ arrives at the origin and leaves the CZ at its exit point, respectively. In the following section we summarize the rules that the CAVs in the CZ have to satisfy so as to maintain a safe flow in the intersection. 

\noindent{\bf Constraint 1} (Rear-End Safety Constraint): Let $i_{p}$ denote the index of the CAV which physically immediately precedes CAV $i$ in the CZ (if one is present). It is required that CAV $i$ conforms to the following constraint:
\begin{equation}
 x_{i_p}(t) - x_i(t) - \varphi v_{i}(t) - \Delta\geq 0,\text{ \ }\forall t\in[ t_{i}^{0},t_{i}^{f}] \label{Safety}%
\end{equation}
where $\varphi$ denotes the reaction time and $\Delta \in \mathbb{R}_{> 0}$ is a given minimum safe distance which depends on the length of these two CAVs.

\noindent{\bf Constraint 2} (Safe Merging Constraint): Every CAV $i$ should leave enough room for the CAV preceding it upon arriving at a MP, to avoid a lateral collision i.e.,
\begin{equation}
\label{SafeMerging}
x_{i_m}(t_i^m) - x_{i}(t_i^m) - \varphi v_{i}(t_{i}^{m}) - \Delta\geq 0,
\end{equation}
where $i_m$ is the index of the CAV that may collide with CAV $i$ at the merging points $m_i=\lbrace 1,...,n_i \rbrace$ where $n_i$ is the total number of MPs that CAV $i$ passes in the CZ. 

\noindent{\bf Constraint 3} (Vehicle limitations): Finally, there are constraints on the speed and acceleration for each $i\in S(t)$:
\begin{equation}
\label{VehicleConstraints1}%
\begin{aligned} v_{min} \leq v_i(t)\leq v_{max}, \forall t\in[t_i^0,t_i^f]\end{aligned}
\end{equation}
\begin{equation}
\label{VehicleConstraints2}%
\begin{aligned} u_{{min}}\leq u_i(t)\leq u_{{max}}, \forall t\in[t_i^0,t_i^f] \end{aligned}
\end{equation}
where $v_{min} \geq 0$ and $v_{max} > 0$ denote the minimum and maximum allowed speed in the CZ respectively, and $u_{min} <0$ $u_{max}>0$ denote the minimum and maximum de(ac)celeration allowed in the CZ respectively. 
%\begin{remark}
    The coordinator finds CAV $i_p$ and CAV $i_m$ for each CAV $i \in S(t)$ from their trajectory and communicates it to CAV $i$. The determination depends on the policy adopted for sequencing CAVs whose relative performance has been studied in \cite{Xu_02}. A common sequencing scheme is the First In First Out (FIFO) policy whereby CAVs exit the CZ in the order they arrive. 
%\end{remark}
\subsection{Decentralized control formulation:}
Under this formulation, each CAV $i$ determines its control policy in a \textit{decentralized manner} based on some objective that includes minimizing travel time and energy consumption, maximizing comfort, etc., governed by the dynamics (\ref{VehicleDynamics}). Expressing energy through $\frac{1}{2}u_i^2(t)$ we use $\alpha\in[0,1]$ as a relative weight between the time and energy objectives, which can be properly normalized by setting $\beta:=\frac{(1-\alpha)\max\{u_{max}^{2},u_{min}^{2}\}}{2\alpha}$ to penalize
travel time relative to the energy cost of CAV $i$. Then, we can formulate an Optimal Control Problem (OCP) as follows:

\begin{equation}\label{eqn:energyobja}
J_{i}(u_{i}(t),t_i^f):=%\underbrace
{\beta(t_{i}^{f}-t_{i}^{0})}%_{J_1,i}
+%\underbrace
{\int_{t_{i}^{0}}^{t_{i}^{f}}\frac{1}{2}u_{i}^{2}(t)dt}%_{J_2,i},
\end{equation} 
subject to Constraints \eqref{VehicleDynamics}-\eqref{SafeMerging}.
%where $\beta:=\frac{\alpha\max\{u_{max}^{2},u_{min}^{2}\}}{2(1-\alpha)}$ is an adjustable weight to penalize travel time relative to the energy cost of CAV $i$. 

\subsection{Trust framework}
\label{trust_framework}
Let $\mathcal{B}$ be a set of indices associated with the behavioral specifications that are used to evaluate the trust of a vehicle. The behavior specifications used in our experiments are listed in \cite{ahmad_02}. For example, conformity to the underlying physical model is a specification that each CAV has to satisfy all the time. For each CAV $i \in S(t), \ \forall t \in [t_i^0,t_i^f]$ the coordinator assigns positive evidence $r_{i,j}(t)$ and negative evidence $p_{i,j}(t)$ for conformance and violation respectively of every specification $j \in \mathcal{B}$ respectively (where $0 \leq r_{i,j}(t) \leq r_{max}, 0 \leq p_{i,j}(t) \leq p_{max}$), which it uses to update the trust $\tau_i(t)$. We define $R_i(t)$ and $P_i(t)$ as cumulative positive and negative evidence for CAV $i$ at time $t$ discounted by trust of other CAVs (if the check involves another CAV, as in \eqref{Safety} and \eqref{SafeMerging}, as they can be untrustworthy). We also define a time discount factor $\gamma \in (0,1)$ as shown below. In addition, we use a non-informative prior weight $h_i$ as in \cite{Cheng_01,Cheng_02}. Let the set of checks for every CAV involving another CAV(s) be denoted by $\mathcal{B}_{a} \subset \mathcal{B}$. The set of other CAVs involved in check $j \in \mathcal{B}_a$ when applied to CAV $i$, is denoted as $S_{i,j}(t) \subseteq S(t)/\{i\}$. %In the case of our considered example of the intersection, such checks are the \textit{control zone rule checks} in \eqref{Safety} and \eqref{SafeMerging} both of which involve another CAV. 
Then, the trust metric %{\color{red}$\tau_i$ is the trust metric? make it clear. Also move (8) to the first equation as it is the most important.} 
is updated as follows: 
%a 3-tuple $\{\tau_i(t),d_i(t),u_i(t)\}$ where 
%is the belief mass 

\begin{equation} \label{trust}
    \tau_i(t) = \frac{R_i(t)}{R_i(t) + P_i(t) + h_i} \ \ \forall i \in S(t)
\end{equation}

\begin{align} \label{evidence}
    R_i(t) = & \gamma R_i(t-1) + \sum_{j \in \mathcal{B} \backslash \mathcal{B}_{a}} r_{i,j}(t) + \sum_{j \in \mathcal{B}_{a}}\prod_{k\in S_{i,j}} \tau_k(t)r_{i,j}(t)  \nonumber \\
    P_i(t) = &\gamma P_i(t-1) + \underbrace{\sum_{j \in \mathcal{B} \backslash \mathcal{B}_{a}} p_{i,j}(t) + \sum_{j \in \mathcal{B}_{a}} \prod_{k\in S_{i,j}} \tau_k(t)p_{i,j}(t)}_{p_i(t)} \nonumber \\
    &\forall i \in S(t), \forall t \in [t_i^0,t_i^f]
\end{align}

Finally, we define a lower trust threshold $\delta \in (0, 1/2)$, and a higher trust threshold $1 - \delta$ for subsequent sections. It is important to emphasize that, in practice, the magnitude of negative evidence is different and significantly higher compared to the magnitude of positive evidence. Note that, every vehicle is deemed untrustworthy i.e. $\tau_i(t_i^0) = 0 \ \forall i$ upon arrival in the CZ. The coordinator updates the trust for every CAV based on the outcome of the behavior specification checks. The zero trust model is used to prioritize safety in our proposed framework.    %This model of trust relationships considers the social aspect, where a single action can cause significant damage to a trust relationship, and recovery from such damage is challenging \cite{Garlichs_01}. 

\section{Threat model}
\label{threat_model}
The adversarial effects of malicious attacks, as highlighted in \cite{ahmad_01}, consist of creating traffic jams across multiple roads due to the cooperative aspect of the control scheme, and, in the worst case, accidents. This warrants making the control robust against these attacks. We consider the attacker models presented in \cite{ahmad_01} in what follows.

\begin{definition}
    (Safe coordination) Safe coordination is defined as the ability to guarantee the satisfaction of \eqref{Safety} and \eqref{SafeMerging} for every CAV $i \in S(t) \ \forall t$ while also conforming to \eqref{VehicleConstraints1} and \eqref{VehicleConstraints2}.
%, to navigate through the CZ without any collision.
\end{definition}

%\begin{definition}
%    (Uncooperative vehicle) We define a CAV $i \in S(t)$ as \emph{uncooperative} if its free-flow speed is abnormally low in the CZ i.e. $v_i(t) \leq v_{low}$ (where $v_{low}$ is considered abnormally low for the CZ), thus worsening traffic throughput.
%\end{definition}

\begin{definition}
    (Adversarial agent) An agent is called \emph{adversarial} if it has one of the following objectives: (i) prevent safe coordination, (ii) reduce traffic throughput.
    %, by introducing cyber-attacks. 
\end{definition}

%Note that adversarial agents introduce attacks with malicious intent, whereas uncooperative CAVs may be slow simply due to reasons like faults or failures. 

\begin{assumption}
\label{accident_assumption}
    Adversarial agents do not collide with other CAVs, nor do they attempt to cause collisions between CAVs and themselves to avoid inflicting loss on themselves.
\end{assumption}

\noindent \textbf{Sybil attack} A single adversarial agent (could be a CAV or attacker nearby the $\mathrm{CZ}$) may spoof one or multiple unique identities and register them in the coordinator queue table as detailed in \cite{ahmad_02}. %We assume at any time $t$, there are two groups of CAVs in CZ: i. Normal CAVs and ii. fake CAVs.
Let $S_x(t)$ and $S_s(t)$ be the set of the indices of normal and fake CAVs in the FIFO queue of the coordinator unit. Therefore at any time $t$, there are $N(t)=\left|S_x(t)\right|+\left|S_{s}(t)\right|$ CAVs which communicate their state and control information to the coordinator. %There can be one or more fake clients/CAVs in the $\mathrm{CZ}$ at any time $t$.
%and, $\left\{\cup_{i = 1}^{|S_x(t)|}\boldsymbol{x}_i(t)\right\}$ denote the set of normal CAV states where $\boldsymbol{x}_i\in\mathbb{R}^2$ includes the position and velocity of CAV $i$, and $\left\{\cup_{i = 1}^{|S_x(t)|}u_i(t)\right\}$ denote the set of control inputs for each of the normal CAVs where ${u}_i\in\mathbb{R}$ is the control input, that are communicated to the RSU by each CAV. Let, $S_{s}(t)$ is the set of their indices in the FIFO queue, $\left\{\cup_{i = 1}^{S_{s}(t)|}\boldsymbol{s}_i(t)\right\}$ is the set of the fake CAV states where $\boldsymbol{s_i} \in \mathbb{R}^2$ includes the position and velocity of fake CAV $i$, and $\left\{\cup_{i = 1}^{|S_x(t)|}\tilde{u}_i(t)\right\}$ is the set of the control inputs of the fake CAVs where $\tilde{u}_i \in \mathbb{R}$ is the acceleration input, that are communicated by the attacker at every time instant.  
A Sybil attack is one where the $S_{s} (t) \subset S(t)$ is a nonempty set that is located in the coordinator queue table, but unknown to the coordinator.
%\end{definition}
%For example, Fig. \ref{fig:intersection} presents a scenario, where there are multiple fake CAVs with indices $S_{s}(t) = \{3,5\}.$
%{\color{red} Give an example corresponding to Fig. 1.}

\begin{assumption}
\label{max_count_assumption}
    There is an upper bound on the maximum number of fake CAVs that an adversary can spoof during a Sybil attack due to resource and energy limitations.
\end{assumption}

\begin{assumption}
    (Bad data detection) The CAVs are equipped with BDDs whereby %with probability $1-\zeta$ ($\zeta$ is small in practice) 
    $\|\boldsymbol{x}_i(t) - 
    \boldsymbol{\hat{x}}_i(t)\|_{\infty} \le \epsilon$, $\forall t, \forall i \in S(t)$ where 
    $\boldsymbol{\hat{x}}_i(t)$ is the measured/estimated state of CAV $i$ at time $t$ and $\|\boldsymbol{x}\|_{\infty}$ is the infinity norm of the state vector.   
\end{assumption}

\noindent\textbf{Stealthy attack} 
An attack is stealthy if $\|\boldsymbol{x}(t) - \boldsymbol{\hat{x}}(t)\|_{\infty} \le \epsilon_1$. Such attacks can be injected through targeting V2I and in-vehicular networks as well as onboard sensing systems.

Specifically, we consider bias injection attacks as defined below:

\textbf{Bias Injection (BI) attack} 
An adversarial agent may attempt to violate safe coordination amongst CAVs, or affect the traffic by targeting one or more CAVs using Person-In-The-Middle attack by adding bias to the data sent by the CAVs to the RSU, or the data sent by the RSU to the CAVs containing state information of the relevant CAVs, or both of them. Let $\boldsymbol{y_i}\left(t\right)$, $i \in S(t)$ be the data (of CAV $i$, or data for CAV $i$ containing the information of the relevant CAVs) injected by the adversary during the attack; and $\boldsymbol{z_i}\left(t\right)$ be the actual data (of CAV $i$ sent to the RSU or data for CAV $i$ containing the information of the relevant CAVs sent by the RSU). Then, during the BI attack, $\boldsymbol{y_i}(t) = \boldsymbol{z_i}(t) + \boldsymbol{g_i(t)}$ where $\|\boldsymbol{g_i}(t)\|_\infty \le \epsilon_1$ is the mapping used by the adversary to generate false data being stealthy.

\begin{assumption} \label{coordinator_assumption}
    We assume that the coordinator is trustworthy i.e., it is not targeted by attacks.
\end{assumption}

%Threat to validity, assumptions does not hold, adding some more threat scenario specifically.

\section{Safe and Resilient Control Formulation using Trust Aware CBFs}
\label{control_coordination}

\subsection{Trust-aware coordination}

The RSU assigns each CAV a unique index based on a passing sequence policy and this information is tabulated and stored according to the assigned indices as shown in Fig. \ref{fig:intersection}. For example, under a FIFO passing sequence the coordinator assigns $N(t)+1$ to a new CAV upon arriving in the CZ. Similarly, each time a CAV $i$ leaves the CZ, it is dropped from the table and all CAV indices larger than $i$ decrease by one.

The coordinator computes and updates the trust metric for each CAV in the CZ as shown in figure \ref{fig:intersection}. The trust metric is incorporated to the selected passing sequence to identify the CAVs any given CAV has to cooperate within the CZ. The cooperation with a CAV involves either constraint \eqref{Safety}, or \eqref{SafeMerging}. According to this method, for every CAV $i \in S(t)$ and for every MP $j \in m_i$, the coordinator identifies the indices of all CAVs that precede CAV $i$ at $j$ based on the selected 
passing sequence until the first CAV whose trust value is greater than or equal to $1- \delta$. This leads to a new set $S_{i,j}(t) \subset S(t)$ containing all the CAV indices identified during the search process. The coordinator follows the same search process for every MP in $m_i$ corresponding to \eqref{SafeMerging}. Therefore, for each CAV $i$, the coordinator identifies $S_{i}^p(t) \subset S(t)$, and $S_{i}^M(t) = \cup_{j \in m} S_{i,j}(t)$ (where $S_{i}^p(t)$ is the set for \eqref{Safety} and $S_{i}^M(t)$ correspond to the set of indices for every MP) and the information is communicated to the CAV. For the example in Fig. \ref{fig:intersection}, note that for CAV 4 we have $i_p = 3$, however since $\tau_3 < 1-\delta$, the search process will continue and return $S_{4,p} = \{1,3\}$. 

\noindent{\textbf{Local sensing}} We also assume that each CAV has a vision-based perception capability defined by a radius and angle pair denoted as $(r,\theta)$, (where $r \in \mathbb{R}^+$,$\theta \in [0,2\pi]$). The incorporation of local sensing into CAV $i \in S(t)$ adds additional constraints of the form \eqref{Safety} to the control problem, besides the constraints corresponding to $S_{i}^p(t)$ and $S_{i}^M(t)$ returned by \textit{trust-based search}. Every CAV $i \in S(t)$ is able to estimate the states of every observed CAV $j$ within its sensing range. CAV $i$ is able to estimate the state of the preceding CAV (if there is one and it is within sensing range) %i.e. $S_{i}^{o} \cap S_{i}^p \neq \emptyset$, %denoted as $\boldsymbol{x}_{i,i_p}^{s}(t)$. Similarly, as CAV $i$ approaches the intersection, it is able to detect CAVs (if there are any) 
and in the vicinity of MPs in its own trajectory; %i.e. $S_{i}^{o} \cap S_{i}^M \neq \emptyset$
in particular, the CAV that will precede $i$ immediately at its next MP should be visible to CAV $i$.

We consider state estimates and communication information  %{\color{red}The above formulation is not noisy estimate but model uncertainty} 
from the coordinator to be noisy as defined below:
\begin{equation}
\label{noisy_state_estimate}
    \boldsymbol{\hat{x}}_{i} (t) = \boldsymbol{x}_{i} (t) + \boldsymbol{w}_{i}(t)
\end{equation}
where $\boldsymbol{w}_{i}(t) = [w_{i}^{(x)}(t), w_{i}^{(v)}(t)]^T$ is random measurement noise with bounded support $\|\boldsymbol{w}_i\|_\infty \le \epsilon_1, \ \forall i \in S(t)$. We can set $\epsilon = \epsilon_1$ (same as the bound for stealthy attacks) to make the controller robust to both noise and stealthy attacks.

%\color{red}(why do you take the infinity norm of the noise vector? The measurement noises for $x_i, v_i$ should be considered separately? Meanwhile, in CBFs, $x_i, v_i$ may appear independently, how would you consider the measurement noise in the CBF constraint?)}. 
%\geq \frac{v_{max}^2}{2|u_{min}|} + \varphi v_{max} \text{ and } 
 %\li{why are we setting $r$ this way? what about $\theta$?}, 
%{\color{red}the following sentence is out of nowhere.}

\noindent \textbf{The OCBF Controller}. This approach uses the OCP formulation in (\ref{eqn:energyobja})
with each state constraint $b_q(\boldsymbol{x}(t)) \geq 0$ mapped onto a new constraint which has the property that it 
implies $b_q(\boldsymbol{x}(t)) \geq 0$ and it 
is \emph{linear} in the control input. The function $b_q(\boldsymbol{x}(t))$ is called a Control Barrier Function (CBF) \cite{Xiao2019}.
We use such CBFs so as to ensure the constraints (\ref{Safety}), (\ref{SafeMerging}), \eqref{VehicleConstraints1} and \eqref{VehicleConstraints2} are satisfied
subject to the vehicle dynamics in (\ref{VehicleDynamics}) by defining $f(\boldsymbol{x}_i(t))=[v_i(t),0]^T$ and $g(\boldsymbol{x}_i(t))=[0,1]^T$. Each of these constraints can be easily written in the form of $b_q(\boldsymbol{x}_{i,j}(t)) \geq 0$, $q \in \lbrace 1,2,3,4 \rbrace$ where $n$ stands for the number of constraints only dependent on state variables $\boldsymbol{x}_{i,j}(t)=[\boldsymbol{x}_i(t),\boldsymbol{x}_j(t)]^T$. 
The general form of the \emph{transformed} CBF-based constraints is:
\begin{equation} \label{cbf_condition}
L_fb_q(\boldsymbol{x}_{i,j}(t))+L_gb_q(\boldsymbol{x}_{i,j}(t))u_i(t)+\kappa_q( b_q(\boldsymbol{x}_{i,j}(t))) \geq 0
\end{equation}
where $L_f,L_g$ are the Lie derivatives of a function along the system dynamics defined by $f,g$ above and
$\kappa_q$ is a class $\mathcal{K}$ function. By combining the OCP formulation in (\ref{eqn:energyobja}) with the CBF-based constraints of the form (\ref{cbf_condition}) instead of the original ones, we obtain the Optimal control with CBFs (termed OCBF) approach detailed in \cite{Xiao_03}.

Finally, the road speed limit can be included as a reference $v_{i}^{ref}(t)$ treated by the controller as a soft constraint using a Control Lyapunov Function (CLF) \cite{Xu_02}
by setting $V(\boldsymbol{x}_i(t))=(v_i(t)-v_{i}^{ref}(t))^2$, rendering the following constraint:
\begin{equation}\label{CLF_constraint}
L_fV(\boldsymbol{x}_i(t))+L_gV(\boldsymbol{x}_i(t))\boldsymbol{u}_i(t)+ c_i V(\boldsymbol{x}_i(t))\leq e_i(t),
\end{equation}
where $e_i(t)$ makes this a soft constraint. The significance of CBFs in this approach is twofold: first, their forward invariance property \cite{Xiao2019} guarantees that all constraints they enforce are satisfied at all times if they are initially satisfied; second, CBFs impose \emph{linear} constraints on the control which is what enables the efficient solution of the tracking problem through a sequence of Quadratic Programs (QPs) thus computationally efficient and suitable for real-time control.

%The control goal is to determine a control law jointly minimizing the travel time and energy consumption subject to constraints \eqref{Safety}, \eqref{SafeMerging}, \eqref{VehicleConstraints1} and \eqref{VehicleConstraints2} for each $i \in S(t)$ governed by the dynamics (\ref{VehicleDynamics}). Expressing energy through $\frac{1}{2}u_i^2(t)$ and normalizing travel time and energy, we use the weight $\alpha\in[0,1]$ to construct a convex combination as follows: The solution is \emph{decentralized} in the sense that CAV $i$ requires information only from CAVs that are relevant to it due to one or both constraints in \eqref{Safety} and \eqref{SafeMerging}.

\subsection{Trust-Aware CBFs}
The choice of the class $\mathcal{K}$ function in (\ref{cbf_condition})
determines the rate at which an agent/CAV reaches the boundary of the safety set. Thus, the choice of this function provides a tradeoff between conservativeness and safety. We can choose a conservative candidate function to prioritize safety by considering all agents to be untrustworthy. However, in view of the available trust metric, we incorporate it in the function with the aim of balancing this tradeoff. The underlying idea is that the degree of conservativeness of a CBF constraint corresponding to a CAV $i$ with respect to CAV $j$ can be adjusted by incorporating the \emph{trust} of CAV $j$, $\tau_j$, in it as shown below:
\begin{equation} \label{trust_cbf_condition}
L_fb_{q}(\boldsymbol{x}_{i,j}(t))+L_gb_{q}(\boldsymbol{x}_{i,j}(t))u_i(t)+\kappa_{q,\tau_j}(b_{q}(\boldsymbol{x}_{i,j}(t))) \geq 0.
\end{equation}%{\color{red}there might be a notation problem of how you use the trust score here. Either clearly state it here or use $\kappa_{q,\tau_j}$ meaning that the class K function depends on the trust score.}
An example for the choice of a class $\mathcal{K}$ function is 
$\kappa_q( b_q(\boldsymbol{x}_{i,j}(t))) = c_{i,j}\tau_j(t) b_{q}(\boldsymbol{x}_{i,j}(t))$, 
where $c_{i,j} \in \mathbb{R}^+$ is a scaling factor.
% \textcolor{red}{Do not make it complicated just write we choose linear and then $... = c_{i,j}\tau_j(t) b_{q}(\boldsymbol{x}_{i,j}(t))$}
\subsection{Robust Trust-Aware CBFs}

In the presence of noisy measurements (estimates) as in \eqref{noisy_state_estimate} the corresponding CBF constraint in \eqref{cbf_condition} can be rewritten as follows due to \eqref{noisy_state_estimate}:
\begin{align} \label{noisy_cbf_condition}
&L_fb_q(\boldsymbol{\hat{x}}_{i,j}(t) - \boldsymbol{w}_{i,j}(t))+L_gb_q(\boldsymbol{\hat{x}}_{i,j}(t) - \boldsymbol{w}_{i,j}(t))u_i(t)+ \nonumber \\ 
&\kappa_{q,\tau_j}(b_{q}(\boldsymbol{\hat{x}}_{i,j}(t) - \boldsymbol{w}_{i,j}(t))) \geq 0.
\end{align}
where $\boldsymbol{w}_{i,j}(t) = [\boldsymbol{w}_i(t), \boldsymbol{w}_j(t)]^T$. For example, the CBF constraint corresponding to \eqref{Safety} is as follows:
\begin{equation}
\small
    v_{i_p}(t) - v_{i}(t) - \varphi u_i(t) - \kappa_{q,\tau_{i_p}}(x_{i_p}(t) - x_{i}(t) - \varphi v_i(t) - \Delta) \ge 0     
\end{equation}
In the presence of noise $\boldsymbol{w}_i(t)$, according to \eqref{noisy_state_estimate} this becomes:
\begin{align*}
    &\hat{v}_{i_p}(t) + w_{i_p}^{(v)}(t) - \hat{v}_{i}(t) - w_{i}^{(v)}(t) - \varphi u_i(t) - \kappa_{q,\tau_{i_p}}(\hat{x}_{i_p}(t) \nonumber\\
    &+ w_{i_p}^{(x)}(t) - \hat{x}_{i}(t) - w_{i}^{(x)}(t) -\varphi \hat{v}_i(t) -  \varphi w_{i}^{(v)}(t)- \Delta) \ge 0
\end{align*}
Obviously, the random noise $\boldsymbol{w}_{i,j}(t)$ is unknown, hence, we use the bound $\epsilon_1$ on the noise to derive the following lemma for the robust trust-aware CBF.

%{\color{red}The above inequality involves the derivative of the noise term $w_{i,j}$. You may want to discuss how you are going to deal with it.}

%{\color{red}It is not clear how you consider the infinity norm in the CBF constraint, and how you do the optimization over the infinity norm in the following CBF constraint. You may want to explicitly show this, for instance, using an example. It is intuitive if you choose the measurement to be bounded by the absolute value of each state variable, respectively.}
\begin{lemma} \label{lem:robust_CBF}
   %{\color{red}The same applies here. The derivative of $w_i$ will be involved.}
   Given a constraint $b_q(\boldsymbol{x}(t))$ associated with the set $\mathrm{C}:=\{\boldsymbol{x}\in \mathbb{R}^n:b_q(\boldsymbol{x})\geq 0\}$ and $\|\boldsymbol{w}_{i,j}\|_{\infty} \leq \epsilon_1$, any Lipschitz continuous controller $u(t)$ that satisfies 
\begin{align} \label{robust_trust_cbf_condition}
&\min_{\{\boldsymbol{w}_{i,j}(t): \|\boldsymbol{w}_{i,j}(t)\|_{\infty} \leq \epsilon_1\}} [L_fb_q(\boldsymbol{\hat{x}_{i,j}}(t) - \boldsymbol{w}_{i,j}(t))]+ L_gb_q(\boldsymbol{\hat{x}_{i,j}}(t) \\ \nonumber    & - \boldsymbol{w}_{i,j}(t))u_i(t) + \kappa_{q,\tau_j}(b_{q}(\boldsymbol{\hat{x}_{i,j}}(t) - \boldsymbol{w}_{i,j}(t)))] \geq 0 
\end{align}
renders the set $C$ forward invariant $\forall t \geq t_{0}$ for the system \eqref{VehicleDynamics}.
\end{lemma}

\begin{proof}
    The satisfaction of \eqref{robust_trust_cbf_condition} guarantees the satisfaction of the constraint \eqref{noisy_cbf_condition} (and \eqref{cbf_condition}) since it is a lower bound for \eqref{cbf_condition} which according to Theorem 1 in \cite{Xiao_03} makes the set $C$ forward invariant $\forall t \geq t_0$ w.r.t \eqref{VehicleDynamics}.
\end{proof} 

Based on the information in the table (as shown in Fig. \ref{fig:intersection}) the coordinator communicates the state information and the trust value of the CAVs in $S_i^p$ and $S_i^M$ corresponding to constraints \eqref{Safety} and \eqref{SafeMerging} respectively to each CAV $i$ in the CZ.

The OCBF problem corresponding to \eqref{eqn:energyobja} is formulated as:
\begin{equation}\label{QP-OCBF}\small
\min_{u_i(t),e_i(t)}J_i(u_i(t),e_i(t)):=\int_{t_i^0}^{t_i^f}\big[\frac{1}{2}(u_i(t)-u_{i}^{ref}(t))^2+\lambda e^2_i(t)\big]dt
\end{equation}
subject to vehicle dynamics (\ref{VehicleDynamics}), the CBF constraints \eqref{robust_trust_cbf_condition}, $\forall q=\{1,...,n\}$ and CLF constraint \eqref{CLF_constraint}. In this approach, $u_i^{ref}$ is generated by solving the \emph{unconstrained} optimal control problem in \eqref{eqn:energyobja} which can be analytically obtained. The resulting control reference trajectory is optimally tracked subject to the constraints.

\subsection{Event-triggered Control}
A common way to solve (\ref{QP-OCBF})) is to discretize $[t_i^0,t_i^f]$ into intervals $[t_i^0,t_i^0+\Delta],...,[t_i^0+k\Delta,t_i^0+(k+1)\Delta],...$ with equal length $\Delta$ and solving (\ref{QP-OCBF}) over each time interval. The decision variables $u_{i,k}=u_i(t_{i,k})$ and $e_{i,k}=e_i(t_{i,k})$ are assumed to be constant on each interval and can be easily calculated at time $t_{i,k}=t_i^0+k\Delta$ through solving a QP at each time step:
%We can approximate the solution of the dynamic optimization problem by rendering it to a QP as given below, and solving it in a time-driven, or, event-driven manner.
\begin{align} \label{QP}
\min_{u_{i,k},e_{i,k}}&[ \frac{1}{2}(u_{i,k}-u_i^{ref}(t_{i,k}))^2+\lambda e_{i,k}^{2}]
\end{align}
subject to the CBF constraints \eqref{robust_trust_cbf_condition}, $\forall q=\{1,...,n\}$, CLF constraint (\ref{CLF_constraint}) and dynamics \eqref{VehicleDynamics}, where all constraints are linear in the decision variables.

This is referred to as the \emph{time-driven} approach. The main problem with this approach is that there is no guarantee for the feasibility of each CBF-based QP, as it requires a small enough discretization time which is not always possible to achieve. Also, it is worth mentioning that synchronization is required amongst all CAVs which can be difficult to impose in real-world applications.
Therefore, to tackle these issues we adopt an \emph{event-triggered} control scheme inspired by \cite{ahmad_03}. 
 %described earlier. %As introduced in \cite{2022XiaoEventAuto}, the key idea is to ensure that the safety constraints are satisfied while the state remains within some bounds and define events which coincide with the state reaching these bounds, at which point the next instance of the QP in (\ref{QP}) is triggered.
Under this scheme, the control for a CAV is updated by solving the QP \eqref{QP} upon the occurrence of any of a predefined set of events (not in the original time-driven fashion) with the goal of ensuring that the state trajectory of the CAV satisfies all the constraints between two consecutive events. We will formulate such a framework for a CAV $i$ w.r.t to another CAV $j$ for a constraint $q \in \{1,\dots,4\}$, corresponding to (\ref{Safety}), (\ref{SafeMerging}) and (\ref{VehicleConstraints1}), which generalizes to every other CAV and constraints. 
Let $t_{i,k}$, and $t_{i,k+1}$ (where $k=1,2,...$), be the time for the $k$-th and $(k+1)$-th event during which vehicle $i$ solves its QP \eqref{QP}. The goal is to guarantee that the state trajectory does not violate any safety constraints within the interval $(t_{i,k},t_{i,k+1}]$. We define $C_{i}$ to be the feasible set of constraints (only dependent on our states \eqref{VehicleConstraints1}) and involving states of another CAV \eqref{Safety},\eqref{SafeMerging}) defined as:
\begin{align}
 \label{event:ci}
    C_{i}\equiv \Bigl\{  &\boldsymbol{x}_{i,j}\in \mathbf{X}^2: b_q(\boldsymbol{x}_{i,j})\geq 0 \text{ and } b_q(\boldsymbol{x}_{i})\geq 0, \nonumber\\ &j \in S_i^P \cup S_i^M \Bigl\}
\end{align} 
%Besides that we define $C_{\bar{i}}$ as following:
%\begin{equation} \label{event:ci_not}
%    C_{i,j}\equiv \Bigl\{  \boldsymbol{x}_{j}\in \mathbf{X}:  b_q(\boldsymbol{x}_{j})\geq 0 \ q \in \lbrace 1,2 \rbrace \Bigl\},
%\end{equation} 
We define a compact convex set on the state space of CAV $i$ at time $t_{i,k}$ such that:
\begin{equation} \label{bound} 
X_i(t_{i,k}) = \Bigl\{\boldsymbol{y}_i \in \mathbf{X}: |\boldsymbol{y}_i - \boldsymbol{x}_i(t_{i,k}) | \leq \boldsymbol{s}_{\boldsymbol{x}_i} \Bigl\}
\end{equation} 
where $\boldsymbol{s}_i \in \mathbb{R}_{>0}^2$ is a parameter vector. Similarly, we define a compact convex set on the trust metric: %{\color{red} this is new in this work, you can emphasize this in the intro and here: an event-triggered trust metric in the CBF-based QP.}:
\begin{equation} \label{trust_bound} 
\mathcal{T}_j(t_{j,k}) = \Bigl\{\tau_j \in [0,1]: |\tau_j - \tau_i(t_{j,k}) | \leq {s}_{\tau_j} \Bigl\}
\end{equation} 
Intuitively, this choice reflects a trade-off between computational efficiency and conservativeness. A larger choice of value makes the controller conservative requiring less frequent control update thus being more computationally efficient, and vice versa. As we use robust CBFs, we need to modify the previously defined sets to adjust the bounds on noisy states as in \eqref{VehicleDynamics}. At first, we define the feasible set %{\color{red}(in the following set, is it defined in terms of $x_{i,j}$ or $\hat x_{i,j}$? It looks ill-defined in its current form, you may emphasize (2) and state it clearly)} 
of constraints as following:
\begin{align} \label{event:ci_hat}
    \hat{C}_{i}\equiv\Bigl\{  &\boldsymbol{\hat{x}}_{i,j}\in \mathbf{X}^2: \min_{\{\boldsymbol{w}_{i,j}:\|\boldsymbol{w}_{i,j}\|_{\infty} \leq \epsilon_1\}}b_q(\boldsymbol{\hat{x}}_{i,j} - \boldsymbol{w}_{i,j})\geq 0 \text{ and } \nonumber \\
    &\min_{\{\boldsymbol{w}_{i}\|\boldsymbol{w}_{i}\|_{\infty} \leq \epsilon_1\}}b_q(\boldsymbol{\hat{x}}_{i} - \boldsymbol{w}_{i})\geq 0 \Bigl\}
\end{align}
%\begin{Proposition}
    Note that $\hat{C}_i \subset C_i$
%\end{Proposition}
because %\begin{equation}
        $b_q(\boldsymbol{\hat{x}}_{i,j} - \boldsymbol{w}_{i,j}) \geq \min_{\{\boldsymbol{w}_{i,j}:\|\boldsymbol{w}_{i,j}\|_{\infty} \leq \epsilon_1\}}b_q(\boldsymbol{\hat{x}}_{i} - w_{i})\geq 0$.
    %\end{equation}
%\end{proof}
The minimum can be derived in closed form %{\color{red} Let's explicitly write down the computation process of the following equation for our own understanding} 
as shown below:
\begin{flalign} \label{event:ci_minimum}
     &\min_{\{\boldsymbol{w}_{i,j}:|\boldsymbol{w}_{i,j}\|_{\infty} \leq \epsilon_1\}}b_q(\boldsymbol{\hat{x}}_{i,j} - \boldsymbol{w}_{i,j}) \nonumber\\
     & = \min_{\{\boldsymbol{w}_{i,j}:|\boldsymbol{w}_{i,j}\|_{\infty} \leq \epsilon_1\}}   \hat{x}_{j} - w_j^{(x)} - \hat{x}_i + w_i^{(x)} - \varphi (\hat{v}_{i} - w_i^{(v)}) - \Delta \nonumber \\ 
     &= \hat{x}_{j} - \hat{x}_i - \varphi v_{i} - \Delta -\epsilon_1 (2 + \varphi) \\ \label{event:ci_minimum_state} 
     & \min_{\{\boldsymbol{w}_{i}:|\boldsymbol{w}_{i,j}\|_{\infty} \leq \epsilon_1\}} b_3(\boldsymbol{\hat{x}}_{i} - \boldsymbol{w}_{i})  \nonumber \\ 
     &=\min_{\{\boldsymbol{w}_{i}:|\boldsymbol{w}_{i,j}\|_{\infty} \leq \epsilon_1\}}\hat{v}_i -  w_i^{(v)} - v_{min} \nonumber \\
     &=\hat{v}_i - v_{min} - \epsilon_1 \\ 
     &\min_{\{\boldsymbol{w}_{i}:|\boldsymbol{w}_{i,j}\|_{\infty} \leq \epsilon_1\}} v_{max} - \hat{v}_i + w_i^{(v)} = v_{max} - \hat{v}_i - \epsilon_1
\end{flalign}
%Then, we define the set of constraints of other CAVs that CAV $i$ has to cooperate with as following:
%\begin{equation} %\label{event:ci_not_minimum}
%    \hat{C}_{i,j}\equiv \Bigl\{  \boldsymbol{\hat{x}}_{j}\in \mathbf{X}:  \min_{\{\boldsymbol{w}_j: \|\boldsymbol{w}_{j}\|_{\infty} \leq \epsilon_1\}}b_q(\boldsymbol{\hat{x}}_{j} - w_{j})\geq 0  \Bigl\},
%\end{equation} 
%The minimum can be derived in closed form similar to \eqref{event:ci_minimum_state}. Finally, we define the state bound as below:
We can similarly define $\hat{X}_j(t_{i,k})$:
\begin{equation} \label{robust_bound} 
\hat{X}_i(t_{i,k}) = \Bigl\{\hat{\boldsymbol{y}}_i \in \mathbf{X}: |\boldsymbol{\hat{y}}_i - \boldsymbol{\hat{x}}_i(t_{i,k}) | \leq \boldsymbol{s}_{\boldsymbol{x}_i} - 2[\epsilon_1, \epsilon_1]^T \Bigl\}
\end{equation} 
where $\boldsymbol{\hat{y}}_i = \boldsymbol{{y}}_i +\boldsymbol{w}_i$. Note that $\hat{X}_i(t_{i,k}) \subset X_i(t_{i,k})$ since
%{\color{red}(It is not clear how you get the first inequality in the below, and how you get the second inequality? Here, note that $\epsilon_1$ is the infinity norm of $w$.Also, what is $w$? are you refering to $w_{i,j}$?)}:
    \begin{equation}
    \begin{aligned}
        |\boldsymbol{y}_i - \boldsymbol{x}_i(t_{i,k})| &=  |\boldsymbol{\hat{y}}_i - \boldsymbol{w}_i - \boldsymbol{\hat{x}}_i(t_{i,k}) + \boldsymbol{w}_i(t_{i,k})| \\ \nonumber 
        &\le |\boldsymbol{\hat{y}}_i - \boldsymbol{\hat{x}}_i(t_{i,k}) + 2\|\boldsymbol{w}\|_{\infty} [1,1]^T | \\ \nonumber
        &= |\boldsymbol{\hat{y}}_i - \boldsymbol{\hat{x}}_i(t_{i,k}) | + 2[\epsilon_1, \epsilon_1]^T
    \end{aligned}
\end{equation}
Thus, $|\boldsymbol{\hat{y}}_i - \boldsymbol{\hat{x}}_i(t_{i,k}) | \le \boldsymbol{s}_{\boldsymbol{x}_i} - 2[\epsilon_1, \epsilon_1]^T \Rightarrow |\boldsymbol{y}_i - \boldsymbol{x}_i(t_{i,k})| \le \boldsymbol{s}_{\boldsymbol{x}_i} $. 
%\end{proof}

Next, we seek a bound and a control law that satisfies the safety constraints within this bound. This can be accomplished by considering the minimum value of each component of \eqref{robust_trust_cbf_condition} as shown next. %{\color{red}(Why the minimization is over the objective at the time $t_{i,k}$ instead of $t$? In (14), all the components are only involved with the time $t$. Please check all the equations from (24) to (26). ALso, in (26), why are there both $b_2$ and $b_q$?)}.
For the first term, let
\begin{equation}\label{minfi}
b^{min}_{q,f_i}(t_{i,k})=\displaystyle\min_{{{\boldsymbol{\hat{y}}_i \in {S}_i({t_{i,k}}) \atop \boldsymbol{\hat{y}}_j \in {S}_{i,j}({t_{i,k}}) }} \atop \{\boldsymbol{w}_{i,j}: \|\boldsymbol{w}_{i,j}\|_{\infty} \leq \epsilon_1\}} L_fb_q\big(\boldsymbol{\hat{y}}_{i,j}(t_{i,k}) - \boldsymbol{w}_{i,j}\big)
\end{equation} 
where $\boldsymbol{\hat{y}}_{i,j}({t_{i,k}}) = [\boldsymbol{\hat{y}}_i({t_{i,k}}), \boldsymbol{\hat{y}}_j({t_{i,k}})]^T,  {S}_i({t_{i,k}}):=(\hat{C}_i \cap \hat{X}_i(t_{i,k}))$, and $ {S}_{i,j}({t_{i,k}}):=  \hat{X}_j(t_{i,k})$. Similarly, we can define the minimum value of the third term in \eqref{robust_trust_cbf_condition}:
\begin{equation}\label{minkappai}
b^{min}_{\kappa_q}(t_{i,k})=\displaystyle\min_{{{\boldsymbol{\hat{y}}_i \in {S}_i({t_{i,k}}) \atop \boldsymbol{\hat{y}}_j \in {S}_{i,j}({t_{i,k}}) } \atop \tau_j \in \mathcal{T}_{j}({t_{i,k}})} \atop \{\boldsymbol{w}_{i,j}: \|\boldsymbol{w}_{i,j}\|_{\infty} \leq \epsilon_1\}} \kappa_{q,\tau_j}\Big(b_q(\boldsymbol{\hat{y}}_{i,j}(t_{i,k}) - \boldsymbol{w}_{i,j})\Big).
\end{equation} 
% is defined as follows:
%\begin{equation}
%    \Bar{S}_i({t_{i,k}}):=\bigl\{\boldsymbol{y}_i \in C_{i,1} \cap S_i(t_{i,k})\bigl\}
%\end{equation}
For the second term in \eqref{robust_trust_cbf_condition}, if it is not constant then the limit value $b^{min}_{2,g_i}(t_{i,k}) \in \mathbb{R}$ can be determined as follows:
\begin{eqnarray}\label{mingi} \small
b^{min}_{q,g_i}(t_{i,k})=\begin{cases}
\displaystyle\min_{{{\boldsymbol{\hat{y}}_i \in {S}_i({t_{i,k}}) \atop \boldsymbol{\hat{y}}_j \in {S}_{i,j}({t_{i,k}})}} \atop \{\boldsymbol{w}_{i,j}: \|\boldsymbol{w}_{i,j}\|_{\infty} \leq \epsilon_1\}} L_gb_q(b_q(\boldsymbol{\hat{y}_{i,j}}(t_{i,k}) - \boldsymbol{w}_{i,j}), \  \\ \textnormal{if}\  u_{i,k} \geq 0\\
\\
\displaystyle\max_{{{\boldsymbol{\hat{y}}_i \in {S}_i({t_{i,k}}) \atop \boldsymbol{\hat{y}}_j \in {S}_{i,j}({t_{i,k}}) }} \atop \{\boldsymbol{w}_{i,j}: \|\boldsymbol{w}_{i,j}\|_{\infty} \leq \epsilon_1\}} L_gb_q(b_q(\boldsymbol{\hat{y}_{i,j}}(t) - \boldsymbol{w}_{i,j})),  \\ \textnormal{otherwise},
\end{cases}
\end{eqnarray}
where the sign of $u_{i,k}$ can be determined by simply solving the CBF-based QP  \eqref{QP-OCBF} at time $t_{i,k}$.

Thus, the condition that can guarantee the satisfaction of a CBF constraint in the interval $\left(t_{i,k},t_{i,k+1}\right]$ is given by
\begin{equation} \label{minCBF}
b^{min}_{q,f_i}(t_{i,k})+b^{min}_{q,g_i}(t_{i,k})u_{i,k}+b^{min}_{\kappa_q}(t_{i,k})\geq 0,
\end{equation}
for $q \in \{1,\dots,4\}$. Note that the minimizations in \eqref{minfi}, \eqref{mingi} and \eqref{minkappai} are simple linear programs whose closed form solution can be easily derived. %is provided in section \ref{appendix}.
In order to apply this condition to the QP \eqref{QP}, we just replace \eqref{robust_trust_cbf_condition} by \eqref{minCBF} as follows:
\begin{align} \label{eq:QPtk} \small
\min_{u_{i,k},e_{i,k}} \big[ \frac{1}{2}\big(u_{i,k}-u_i^{ref}(t_{i,k})\big)^2+\lambda e_{i,k}^{2}\big]\ \ \textnormal{s.t.} \ \  \eqref{CLF_constraint},\eqref{minCBF}, \eqref{VehicleConstraints2}
\end{align} 

Finally, we can determine $t_{i,k+1}$, the next time that a solution of the QP \eqref{eq:QPtk} must be solved, as follows:
\begin{align} \label{events}
t_{i,k+1}=\min \Big\{ & t>t_{i,k}:\vert\hat{\boldsymbol{x}}_i(t)-\hat{\boldsymbol{x}}_i(t_{i,k})\vert\ge\boldsymbol{s}_{\boldsymbol{x}_i} - 2[\epsilon_1,\epsilon_1]^T \\ \nonumber
&\text{or} \ \ \vert\hat{\boldsymbol{x}}_{j}(t)-\hat{\boldsymbol{x}}_{j}(t_{i,k})\vert\ge\boldsymbol{s}_{\boldsymbol{x}_j} - 2[\epsilon_1,\epsilon_1]^T, \ \forall j \\ \nonumber 
&\text{or} \ \ \vert\tau_{j}(t)-\tau_{j}(t_{i,k})\vert\ge\boldsymbol{s}_{\tau_j} \forall j\Big\}, \ \ \ t_{i,1}=0
\end{align}
%where $t_{i,1}=0$. %{\color{red}(add the measurement uncertainty to the above triggering condition for the next $t_{k-1}$, i.e., use the measurement bounds instead of the exact state values in the above equation).}

The following theorem formalizes our analysis by showing that if new constraints of the general form \eqref{minCBF} hold, then our original CBF constraints \eqref{trust_cbf_condition} also hold. The proof follows the same lines as that of a more general theorem in \cite{2022XiaoEventAuto}.

%\textcolor{red}{use this paper notation in the theorem you forgot to change the notation}
\begin{theorem}\label{as:1} Given a CBF $b_q(\boldsymbol{x}_{i,j}(t))$ with relative degree one, let $t_{i,k+1}$, $k=1,2,\ldots$ be determined by \eqref{events} with $t_{i,1}=0$ and $b^{min}_{q,f_i}(t_{i,k})$, $b^{min}_{\gamma_q}(t_{i,k})$, $b^{min}_{q,g_i}(t_{i,k})$ obtained through \eqref{minfi}, \eqref{minkappai}, and \eqref{mingi}. Then, any control input $u_{i,k}$ that satisfies \eqref{minCBF} for all $q \in \{1,\dots, 8\}$  within the time interval $[t_{i,k},t_{i,k+1})$ renders the set $\hat{C}_i$ and therefore $C_i$ forward invariant for the dynamic system defined in (\ref{VehicleDynamics}).
\end{theorem}
\begin{proof}
    The satisfaction of \eqref{minCBF} satisfies the constraint \eqref{trust_cbf_condition} which in turn satisfies \eqref{robust_trust_cbf_condition} which makes $C_i$ forward invariant based on Lemma \ref{lem:robust_CBF}.
\end{proof}

\begin{corollary} \label{cor:1}
    The satisfaction of \eqref{minCBF} corresponding to \eqref{Safety}, \eqref{SafeMerging} and \eqref{VehicleConstraints1}, subject to \eqref{VehicleConstraints2} guarantees satisfaction of the constraints \eqref{Safety} and \eqref{SafeMerging} for $\|\boldsymbol{x}_i(t) - \boldsymbol{\hat{x}}_i(t)\|_{\infty} \le \epsilon_1 \ \forall t, \forall i \in S(t)$.
\end{corollary}

\begin{proof}
    The satisfaction of \eqref{minCBF} makes the set $\hat{C}_i$ and correspondingly set $C_i$ for \eqref{Safety}, \eqref{SafeMerging} and \eqref{VehicleConstraints1} forward invariant from Theorem \eqref{as:1} for any $\|\boldsymbol{w}_i(t)\|_{\infty} = \|\boldsymbol{x}_i(t) - \boldsymbol{\hat{x}}_i(t)\|_{\infty} \le \epsilon_1$ guaranteeing their satisfaction $\forall t, \forall i \in S(t)$.  
\end{proof}

\begin{corollary}
    The trust-based coordination in conjunction with control using robust trust-aware CBFs guarantees safe navigation of CAVs against Sybil attacks and Stealthy attacks.
\end{corollary}

\begin{proof}
    The trust-based search guarantees safe coordination against Sybil attacks as proved in \cite{ahmad_02} (in Theorem 1) and in conjunction with robust trust-aware CBFs from Corollary \ref{cor:1} makes our control and coordination framework safe against Sybil and Stealthy attacks. 
\end{proof}

\section{Attack Detection and Mitigation}
\label{mitigation}
%{\color{red} This section should go before the last section?}
%{\color{red} It there any connection between this section and the last section? In other words, is a spoofed CAV from two sources of attack: (i) fake info. (ii) existence of a CAV? Does the last section address (i) and this section addresses (ii)? Let's make the connection very clear here.}
Our proposed robust control scheme offers provably safe coordination against adversarial attacks. However, there are scenarios where attackers may target the network performance by causing traffic holdup. This is possible with Sybil attacks, as illustrated in \cite{ahmad_01}, necessitating attack mitigation besides safety guarantees. The problem of detection involves the identification of adversarial (or, spoofed) CAVs accurately and mitigation can be defined as reestablishing the normal cooperation in the network close to what it would be in the ideal scenario without any attack. Resilience is necessary to ensure safe coordination until the attack is detected and in the presence of any false identification of adversarial (or, spoofed) CAVs. In this section, we present our proposed mitigation framework based on the trust framework with the aforementioned objective.

\subsection{Determination of Fake CAVs} Initially, every CAV is considered untrustworthy (i.e., $\tau_i(t_i^0) = 0$). Upon arrival in the CZ, the coordinator monitors the trust for each CAV and, if it detects any CAV $i \in S(t)$ s.t. $\tau_i(t) \leq 1 - \delta$ and $\tau_i(t) \leq \tau_i(t-1)$, it initiates an observation window for that particular CAV of length $\eta$. If the trust for CAV $i$ is non-increasing and stays below the threshold of $1 - \delta$ during the observation window then the coordinator proceeds to the mitigation step.

\subsection{Robust Mitigation}
The most trivial strategy that can be adopted is to rescind cooperation with the fake CAVs; however, it is essential to note that our  framework can output false positives (although highly unlikely if the priorities of the behavioral specifications are chosen as mentioned in \cite{ahmad_02}). 
Therefore, we offer a soft mitigation scheme; we call it ``soft'' because it is a passive scheme that relies on the local sensory information of the CAVs. This will become apparent in the remainder of the section. We define a \emph{rescheduling zone} in the CZ of length $L_1$ as shown in Fig. \ref{fig:intersection}. It has been shown that any passing sequence can be rescheduled in this area in \cite{ahmad_03}. Then, we present the following definitions.

\begin{definition}
    (Explicitly constrained agent) An agent $i$ is called \emph{explicitly constrained} by an agent $j$ at time $t$ if it has a constraint directly involving states of agent $j$ at that time.
\end{definition}

\begin{definition}
    (Implicitly constrained agent) An agent $i$ is called \emph{implicitly constrained} by an agent $j$ at time $t$ if there is any other agent $k$ in the environment constrained by $j$, which constrained agent $i$.
\end{definition}

%We make a few observations to motivate the idea of our presented mitigation scheme. Firstly, the re-sequencing zone facilitates rescheduling similar to the concept of trust based rescheduling. However, it is important to note that, in this scenario, every CAV apart from the fake CAV should have equal priority and the fake CAV should have the least priority. This is because we don't want to penalize the other CAVs apart from the fake CAV by changing their index in the passing sequence through rescheduling.Based on the observation we present a ILP in \eqref{resequencing_mitigation} similar to \eqref{lane_priority_resequencing} except the objective is to

We mitigate the effect of fake CAVs by unconstraining the CAVs that are explicitly constrained by them (including the physically following CAVs if they are within their perception range and do not actually see any vehicle ahead) by solving the Integral Linear Program (ILP) defined below. Let the set of the ordered indices of detected fake CAVs that we want to mitigate be denoted as $S_f(t)$. We define the index $k_{min}=\min S_f\left(t\right)$ as the index of the first (fake) CAV in the queue to re-sequence from and $S_+(k_{min}) = \{k_{min},\dots,N(t)\}$. Then, the ILP is formulated as follows:
\begin{align}
    \max_{i \in S_f(t)} &\sum_{i \ \in S_f(t)} {a_i}  \label{ILP} \\
    &a_j - a_k \geq \nu, \ \forall k \in \bar{S}_f(t) \cap S_{j}^p(t), \nonumber \\ 
    & \text{and } j \in S_+(k_{min})   \label{rearend_ILP} \\
    &a_j - a_k \geq \nu,  j \in S_+(k_{min}), k \in S_{j}^M(t)   \label{merging_ILP} \\
    &a_j \neq a_k \ j,k \in S_+(k_{min}) \nonumber \\
    %&a_j-a_i \geq -\varepsilon+ N(t)y_{i,j}  \ \forall i \in [s,\dots, N(t)] \nonumber \\
    %&j \in [s, \dots, N(t)] \nonumber \\
    %&a_j-a_i \geq\varepsilon-(1-y_{i,j}) N(t) \ \forall i\in [s,\dots, N(t)] \nonumber \\
    %&j \in [s, \dots, N(t)] \nonumber \\
    &\left\{a_{k_{min}}, \ldots, a_N(t)\right\} \in S_+(k_{min}); \nu \geq 1 \label{resequencing_mitigation}
    %&k_{new}=a_k \ \forall k_{new} \in [k_{min},\dots, N(t)] 
\end{align}
where \eqref{rearend_ILP} correspond to constraint \eqref{Safety}, \eqref{merging_ILP} correspond to constraint \eqref{SafeMerging}, $\{a_{k_{min}},\dots, a_{N(t)}\}$ are the new indices of the CAVs in $S_+(k_{min})$. 
 
Based on the above definitions we now outline the scenarios that are of importance to us and derive an approximate solution of \eqref{ILP} for them.
\begin{enumerate}
    \item \textbf{No CAVs are constrained by  CAVs in $S_f(t)$}: In this case, the solution of \eqref{ILP} will reschedule the CAVs starting from index $k = \min S_f(t)$ in $S_f(t)$ by moving them at the end of the queue and move the remaining CAVs with original index $i \geq k \text{ and } i \notin S_f(t)$ ahead in the queue to fill their places in their current order. This process will be repeated $\forall k \in S_f(t)$.
    \item \textbf{There are CAVs in $S_f(t)$ which physically precede another CAV in the CZ}: First, let us consider CAV $k \in S_f(t)$ and $j$ is the index of physically immediately following CAV, and let $S_j^c(t) \subseteq S(t)$ be the set of CAVs explicitly and implicitly constrained by $j$. At first, the CAVs with indices between $k$ to $j-1$ are moved ahead in the queue by incrementing their index by 1, then, we set $k \leftarrow j-1$ where $j-1 > k$. The reason for moving $k$ down the queue up to $j-1$ is because $k$ can be a real CAV which has been falsely identified as a fake CAV. Finally, remove $\{j-1,j\} \cup S_j^c(t)$ from the queue, rearrange the queue by incrementing the indices of the remaining CAVs appropriately in the queue, and add $\{j-1,j\} \cup S_j^c(t)$ in the queue. Then, repeat the process for the remaining CAVs in $S_f(t)$. Finally, update $S_f(t)$ accordingly.
    The final step is needed to move any CAVs $k>j-1$ that are not explicitly constrained, or implicitly constrained by the immediately preceding CAV of CAV $j+1$ ahead of CAV $j-1$ in the queue.    \label{shuffle}
\end{enumerate}
    %intended to move those CAVs that are unconstrained by the $physically \ following$ CAVs $\forall k \in S_f(t)$. This is 

Observe that the rear-end constraints are excluded for the CAVs that are physically immediately following any CAV $k\in S_f(t)$ in \eqref{resequencing_mitigation} to allow CAVs that are physically immediately behind the CAVs in $S_f(t)$ to overtake them \textit{only if} they are not visible when within sensing range. This is necessary to guarantee safety for FP cases which will be described later. 
%To do so, upon completion of the rescheduling as per the relevant scenario described previously, we modify the CBF constraint in \eqref{noisy_cbf_condition} as follows. 
%\begin{align} \label{noisy_modified_cbf_condition}
%&L_fb_{k+1,q}(\boldsymbol{\hat{x}}(t) - \boldsymbol{w}(t))+L_gb_{k+1,q}(\boldsymbol{\hat{x}}(t) - \boldsymbol{w}(t))u_{k+1}(t) \nonumber \\ 
%&+(\gamma+\rho)( b_{k+1,q}(\boldsymbol{\hat{x}}(t) - \boldsymbol{w}(t))) \geq 0. \nonumber \\
%&\text{where} \ \rho \in \mathbb{R}_{+} \ if \ (k \in S_f(t) \ \land \ k \text{\ is not in sensing range})
%\end{align}

Moreover, observe that, for CAV $k \in S_f(t)$, upon rescheduling, the index of its immediately following CAV will become $k+1$. Once within the sensing range of CAV $k+1$, if CAV $k$ is not visible, it changes its control in \eqref{QP} by removing the CBF constraint corresponding to CAV $i$ to complete the overtake. The coordinator detects the overtake completion by checking the satisfaction of the inequality in \eqref{overtake_check}, upon which it completes the final step of the \eqref{ILP} by swapping the indices of CAV $k$ and $k+1$ with each other. This step is performed $\forall k \ \in S_f(t)$ and repeated by following the scenarios mentioned previously (i.e., the solution of ILP) until all fake CAVs reach the end of the queue.
\begin{equation}
 \hat{x}_{i}(t) - \hat{x}_{i_{p}}(t) - \varphi \hat{v}_{i_{p}}(t) - \Delta\geq 0\text{ \ } \label{overtake_check}%
\end{equation} 

For FP cases, notice that for any $(j-1) \in S_f(t)$,  $j$ is the CAV physically preceding it and every CAV $j^+ > j$ that is not explicitly or implicitly constrained by $j$ is scheduled ahead of them after the first iteration of the algorithm. There will be no further rescheduling for $j-1$, i.e., there will be no vehicles overtaking it in the same road.

\begin{lemma}
    The proposed mitigation scheme guarantees safety for real CAVs even if they are falsely identified as fake CAVs due to a Sybil attack.
\end{lemma}

\begin{proof}
    In the rescheduling zone, any real CAV $i \in S(t)\backslash S_f(t)$ only overtakes a CAV $k \in S_f(t)$ if it does not observe $k$ through its local perception. Similarly, any CAV $i \in S(t)\backslash S_f(t)$ only ignores the CBF condition in its control and jumps ahead of a CAV in $S_f(t)$ in the intersection if it does not observe it through its local vision. This makes our proposed mitigation scheme soft (or passive) and guarantees safety for false positive cases i.e., real CAVs which have been misidentified as fake CAVs. 
\end{proof}
   
The fake CAVs are removed from the coordinator queue in one of two ways: (i) the attacker stops sending information about a fake CAV, and (ii) the fake CAV leaves the CZ. 

\begin{figure*}[h]
\begin{center}
 \includegraphics[scale = .1]{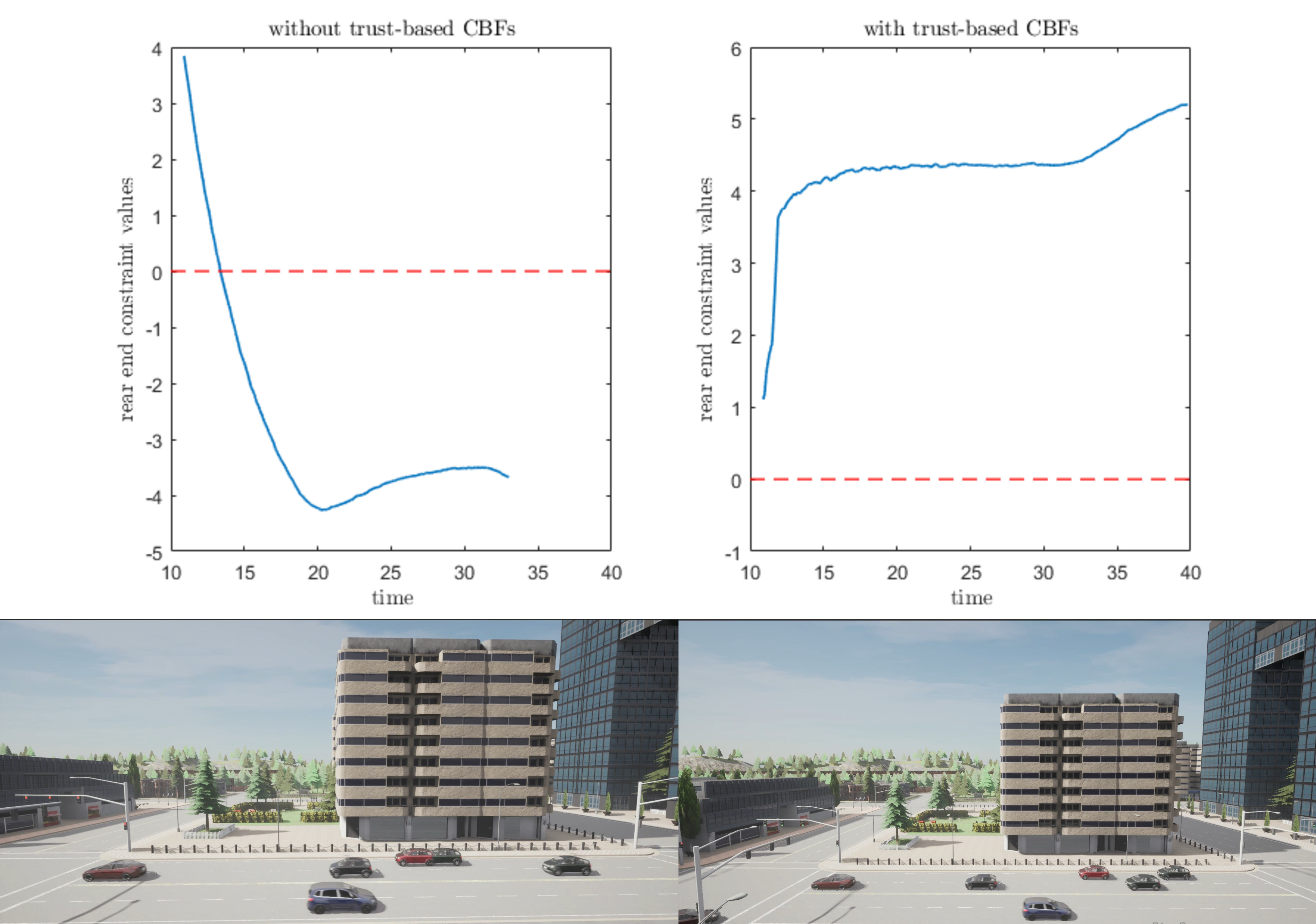}
\caption{Results illustrating the merit of our proposed robust trust-aware event-triggered control scheme. The result was generated by simulating an attack scenario combining BI attack with Sybil attack. As can be seen, the framework in \cite{ahmad_02} results in safety violation (left) which is prevented by our proposed robust trust-aware event- triggered control scheme. The images shown above are from CARLA simulations.} %
\label{fig:accident}%
\end{center}
\end{figure*}

\begin{figure*}[t]
\begin{center}
 \includegraphics[scale = 0.7]{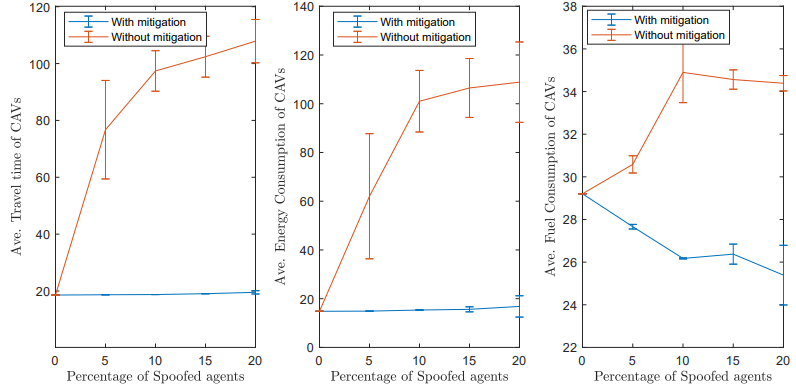}
\caption{The values of average travel time, average energy, and average fuel consumption for real CAVs for different proportions of fake CAVs over 5 runs with and without our proposed mitigation scheme.} %
\label{fig:mitigation_result}%
\end{center}
\end{figure*}

\begin{figure*}[t]
\begin{center}
 \includegraphics[scale = .95]{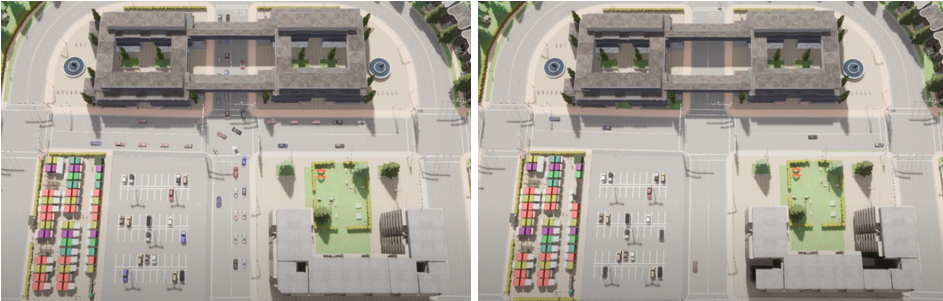}
\caption{The figure shows the performance of the network during a Sybil attack containing six spoofed CAVs without (left) and with (right) our proposed attack mitigation scheme. The picture was taken after 1 minute of running the simulation. The spoofed CAVs were located in three of the eight lanes.} %
\label{fig:mitigation_result_2}%
\end{center}
\end{figure*}

%\begin{figure}[htb]
%\begin{center}
% \includegraphics[scale = 0.5]{Positions.pdf}
%\caption{Illustration of variation of position of CAVs in the CZ. \color{red}The font is too small in the figure.} %
%\label{fig:Positions}%
%\end{center}
%\end{figure}  

\section{Simulation Results}
\label{results}
In this section, we present simulation results for the application of our proposed trust-aware robust CBF based event-triggered control and coordination scheme, including results for mitigation applied to various attacks mentioned in Section \ref{threat_model}. Throughout, we set $\delta = 0.1$ and $\eta = 40$. The positive and negative evidence magnitudes for the tests in the order they are mentioned in Section \ref{trust_framework} are: $r_i(t) = [0.6,0.6,0.6,0.6]$ and $p_i(t) = [1000,100,50,1] \ \forall i \in S(t)$ and $\forall t$. The intersection dimensions are: $L= 400\textnormal{m}$, $A = 300\textnormal{m}^2$; and the remaining parameters are $\varphi = 1.8\textnormal{s}$,  $\Delta = 3.78 \textnormal{m}$, $\beta_1= 1, u_{\max} = 4.905 \textnormal{m/s}^2, u_{\min} = -5.886\textnormal{m/s}^2, v_{\max} = 108 \textnormal{km/h}, v_{\min} = 0 \textnormal{km/h}$. Finally, we also used a realistic energy consumption model from \cite{Kamal} to supplement the simple surrogate $L_2$-norm ($u^2$) model in our analysis:
$f_{\textrm{v}}(t)=f_{\textrm{cruise}}(t)+f_{\textrm{accel}}(t)$ with
\begin{align*}
    f_{\textrm{cruise}}(t) &= \omega_0+\omega_1v_i(t)+\omega_2v^2_i(t)+\omega_3v^3_i(t),\\
    f_{\textrm{accel}}(t) &=\big(r_0+r_1v_i(t)+r_2v^2_i(t)\big)u_i(t).
\end{align*}
where we used typical values for parameters $\omega_1,\omega_2,\omega_3,r_0,r_1$ and $r_2$ as reported in \cite{Kamal}. The simulation was done in Sumo and Carla, where we used Sumo to generate various traffic scenarios and Carla to validate and evaluate the performance of our proposed schemes. 

\noindent \textbf{Trust-aware CBFs:} We present results comparing trust aware robust CBFs with ordinary CBFs using the event-triggered control framework. The results are summarized in Table \ref{Table_event}, containing simulations for 30 vehicles with a Poisson traffic arrival process whose rate was set to 400 vehicles/hour. In the ordinary CBF case, the class $\mathrm{K}$ function is set to be linear in its argument: $\kappa_q=\kappa^{\prime}_q.b_q(.) $ where $\kappa^{\prime}_q= 0.1$. $\alpha$ is similar to as defined in \eqref{eqn:energyobja}.  We can see the benefits of incorporating the trust metric into CBFs, as there is a mixture of low-trust and high-trust vehicles. As can be seen, integrating trust makes the CBFs less conservative reducing the average travel times of the CAVs in the network and increasing average acceleration, thus improving the throughput of the network. Finally, we notice that this also improves the average fuel consumption of the vehicles in the network.

% \begin{table}[hb]
% \caption{Event-triggered control performance comparison with and without trust based CBF %{\color{red}what is energy? What is the percentage of spoofed CAVs? Can you relate this table to the plots in Fig. 2? How would you quantify the safety with/without the trust-based CBFs?}
% }
% \label{tab:trust_no_trust_CBF}
% \begin{tabular}{|l|l|l|l|}
% \hline
%  & Average time (s) & Average $\frac{1}{2}u^2$ & Average fuel \\
% & &  &      consumption \\ \hline
% \begin{tabular}[c]{@{}l@{}}Without trust \\ based CBF\end{tabular} & 24.32 & 10.76           & 14.16          \\ \hline
% With trust-based CBF   & 21.57   & 14.77   & 18.56   \\ \hline
% \end{tabular}
% \end{table}

\begin{table}[h]\scriptsize
        \caption{Event-triggered control performance comparison with and without trust based CBF}
        \centering
        \begin{tabular}{|c|c|c|c|}
            \cline{1-4}
             &Item & CBF with trust & CBF without trust\\
            \cline{2-4}
            \hline
        \multirow{3}{*}{\makecell{$\alpha=0.9$ }} & Ave. Travel time & \textcolor{green}{25} & \textcolor{red}{30.10} \\
        \cline{2-4}
        & Ave. $\frac{1}{2} u^2$ & \textcolor{green}{1.2} & \textcolor{red}{3.10}  \\
        \cline{2-4}
        & Ave. Fuel consumption &\textcolor{green}{17.73} & \textcolor{red}{18.50} \\
        \cline{2-4}
        \hline
        \multirow{3}{*}{\makecell{$\alpha=0.75$ }}  & Ave. Travel time & \textcolor{green}{22.58} & \textcolor{red}{27.70} \\
        \cline{2-4}
        & Ave. $\frac{1}{2} u^2$& \textcolor{red}{3.80} & \textcolor{green}{3.16} \\
        \cline{2-4}
        & Ave. Fuel consumption & \textcolor{green}{17.36} & \textcolor{red}{18.55}  \\
        \cline{2-4}
        \hline
                \multirow{3}{*}{\makecell{$\alpha=0.6$ }}  & Ave. Travel time & \textcolor{green}{22.2} & \textcolor{red}{27.59} \\
        \cline{2-4}
        & Ave. $\frac{1}{2} u^2$&  \textcolor{red}{5.65}& \textcolor{green}{4.75} \\
        \cline{2-4}
        & Ave. Fuel consumption & \textcolor{green}{17.49} & \textcolor{red}{18.65}\\
        \cline{2-4}
        \hline            
        \end{tabular}
        \label{Table_event}
\end{table}

%\begin{figure}[htb]
%    \begin{center}
%     \includegraphics[scale = 0.5]{TrustValue.pdf}
%    \caption{Illustration of variation of trust values of fake CAVs.} 
%    \label{fig:TrustValue}%
%    \end{center}
%\end{figure} 

%\begin{figure}[htb]
%    \begin{center}
%     \includegraphics[scale = 0.5]{Positions.pdf}
%    \caption{Illustration of variation of trust values of fake CAVs.} 
%    \label{fig:Positions}%
%    \end{center}
%\end{figure} 

\noindent\textbf{Bias Injection Attack} In order to highlight the robustness of our scheme against stealthy attacks and noise/estimation uncertainties we simulated an attack scenario by combining Sybil attack with BI attack. We compare our framework against the non-robust framework proposed in \cite{ahmad_02} and the results are shown in Fig.~\ref{fig:accident}. As can be seen, the attack violates constraint \eqref{Safety} as shown in the plot of the constraint value (top left) which becomes negative due to the attack. This results in safety violation resulting in collision as shown in the image (on the left). On the other hand our proposed framework ensure safe coordination as can be verified from the plot and the image (on the right).

\noindent\textbf{Mitigation:} The ultimate goal of having mitigation in place is to avoid accidents and minimize the effects of attacks on the performance of the traffic network (i.e., average travel time, average energy consumption, and average fuel consumption). We present our empirical results in Fig. \ref{fig:mitigation_result} by injecting different proportions of fake CAVs during the attack and for each scenario performing 5 runs whose average and standard deviation are shown in the plots. We considered the strategic attacker model presented in \cite{ahmad_01}. It is important to note that this model assumes that the attacker has no access to the RSU. We varied the location of the spoofed CAVs, their initial states, and the proportion of spoofed CAVs across the runs. As can be seen, with our proposed mitigation scheme the average travel time was reduced to almost the same value as the scenario with no attack, thus validating the efficacy of the mitigation scheme in maintaining network performance. In addition, the average energy was also reduced to almost what it was without an attack. Moreover, we notice that the average fuel consumption improves with our proposed mitigation scheme. 

Additionally, we provide a simulation scenario from CARLA during a Sybil attack in Fig. \ref{fig:mitigation_result_2}. The two figures show the network performance with and without our proposed mitigation scheme after 1 min. of starting the simulation. As can be seen, the absence of mitigation causes traffic holdup which is eased with our proposed mitigation scheme. 

%\noindent\textbf{False Positive}: %The main reason behind introducing soft mitigation instead of hard mitigation is detecting a FP case in which a real CAV falsely gets detected as the fake car and the mitigation gets triggered for it. 
%In the given example in \eqref{fig:intersection}, we fabricated a scenario where a real CAV $5$ loses trust and eventually $\tau_5 \leq 1- \delta$ as can be seen in Fig. \eqref{fig:TrustValue} (in orange color). As a result, mitigation is triggered resulting in CAV $6$ and CAV $7$ (located in re-sequencing zone) jumping ahead of it since they are in different lanes, as can be corroborated from Fig. \eqref{fig:Positions}.  CAV 6 and 7 (going straight) have common MPs as CAV 5 (turning left). Hence, during mitigation, they will jump ahead of CAV 5, and followingly CAV 5 stays safe relative to them at the common MPs. However, CAV $8$ doesn't overtake it since it detects it through its vision, and stays behind CAV $5$ thus guaranteeing safety, \eqref{fig:Positions}.

\noindent\textbf{False positive case.} As mentioned, our choice of trust framework does not result in false positive cases. However, as our proposed method is invariant to the specific choice of the trust framework, we conducted experiments to analyze scenarios when a real CAV gets falsely identified as spoofed due to a poorly chosen trust framework. We conducted our experiments for various degrees of accuracy of the onboard vision system. For each scenario, we ran 100 experiments and computed the percentage of safe scenarios, with results shown in Fig. \ref{fig:vision_acc}. The experiments were run under different traffic conditions by varying the location of the falsely identified CAV (as spoofed) at the intersection for various values of states for the preceding vehicle(s). An experiment was deemed ``safe'' if there were no collisions between real CAVs upon triggering mitigation. Our experiments show that we can guarantee safety with $95\% - 99.96\%$ accuracy when the accuracy of the onboard object detection pipeline varies from $85\% - 95\%$.    

\begin{figure}[ht]
\begin{center}
 \includegraphics[scale = 0.6]{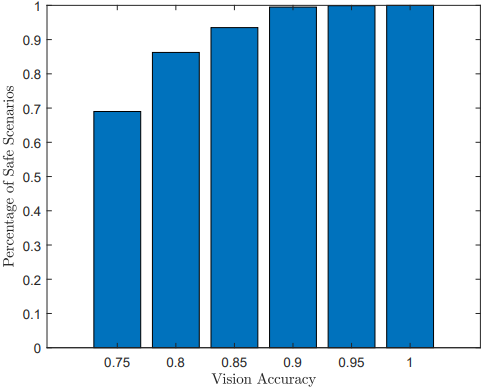}
\caption{Percentage of safe scenarios over 100 runs for different degrees of accuracy of the onboard vision system.} %
\label{fig:vision_acc}%
\end{center}
\end{figure}

\section{Conclusion}
\label{conclusion}
We have addressed Stealthy attacks namely Bias Injection attacks and Sybil attacks on cooperative control of a network of CAVs in a conflicting roadway. We propose decentralized event-triggered control framework using robust trust-aware CBFs. Our proposed framework provides twofold benefits. Firstly, it guarantees provably safe coordination in the presence of adversarial attacks. Secondly, CBFs require choosing a class $\mathcal{K}$ function that inherently poses a tradeoff between conservativeness and safety. We combine trust metric associated to each CAV to balance this tradeoff where the trust of each CAV is intended to reflect the normalcy of a CAV. It is important to note that our proposed framework is invariant to the specific implementation of the trust framework. In addition, we propose a soft attack mitigation scheme to restore normal operation of the road network in the presence of attacks. Our proposed mitigation scheme can guarantee safety coordination against false positive cases. Our simulation results acquired using SUMO and CARLA highlights the merits of our proposed control and coordination scheme and validates their efficacy. In future works, we will extend our work by considering sensor attacks in particular attacks on the Vision, Radar and LIDAR systems along with attacks on in-vehicular network.

\bibliographystyle{IEEEtran}

\end{document}